\documentclass[11pt]{pjm}
\usepackage[dvips]{graphicx}
\title[Algebraic structure of quasiradial solutions]
{Algebraic structure of  quasiradial solutions to the
$\gamma$-harmonic equation}

\author[Vladimir G. Tkachev]{Vladimir G. Tkachev}
\address{Mathematical Department \\ Volgograd State University\\ 2-ja Prodolnaja 30 Volgograd\\ 400062 Russia}
\email{vladimir.tkachev@volsu.ru}
\urladdr{http://www.math.kth.se/\~{}tkatchev}

\thanks{The author was supported by Russian President grant
for young doctorates no.~ 00-15-99274, grant RFBR no.~03-01-00304
and grant of MHE no.~ E~02-11.0-39.}

\newtheorem{thm}{Theorem}[section]
\newtheorem{prop}[thm]{Proposition}
\newtheorem{lem}[thm]{Lemma}
\newtheorem{cor}[thm]{Corollary}

\theoremstyle{definition}

\newtheorem{defn}[thm]{Definition}

\newtheorem{conj}[thm]{Conjecture}

\theoremstyle{remark}

\date{Received date / Revised version date}
\def\aeq{{\uxsq\uxx + 2 \ux \uy \uxy + \uysq \uyy}}
\def\ux{{u_{x}}}
\def\uy{{u_{y}}}
\def\uxx{{u_{xx}}}
\def\uxy{{u_{xy}}}
\def\uyy{{u_{yy}}}
\def\uxsq{{u_{x}^{2}}}
\def\uysq{{u_{y}^{2}}}

\def\Com#1{\mathbb{C}^{#1}}
\def\R#1{\mathbb{R}^{#1}}
\def\Z#1{\mathbb{Z}^{#1}}

\def\N#1{\mathbb{N}^{#1}}
\def\Q{\mathbb{Q}}

\def\scal#1#2{\langle #1; #2 \rangle}

\def\t{s}

\DeclareMathOperator{\Div}{div} 
 \DeclareMathOperator{\re}{Re }
\DeclareMathOperator{\im}{Im }
\def\g{\gamma}

\begin{document}
%\usepackage[dvips]{graphicx}

%\usepackage{refcheck}

%\usepackage{amsmath}
%\usepackage{amssymb}

%\usepackage{esdiff}
%\usepackage{refcheck}

%%*****************************************************

%\usepackage{refcheck}
%%*****************************************************

\begin{abstract}
We obtain an explicit representation for quasiradial
$\gamma$-har\-mo\-nic functions, which shows that these functions have
essentially algebraic nature. In particular, we give a complete
description of all $\gamma$ which admit algebraic quasiradial
solutions. Unlike the cases $\g=\infty$ and $\g=1$, only finitely
many algebraic solutions is shown to exist for any fixed $|\g|>1$.
Moreover, there is a special extremal series of $\g $ which exactly
corresponds  to the well-known ideal $m$-atomic gas adiabatic
constant $\g=\frac{2m+3}{2m+1}$.
\end{abstract}

\maketitle

%\keywords{$\gamma$-harmonic functions, quasiradial solutions,
%algebraic functions }

%%%%%%%%%%****************************************************

\section{Introduction}

%\subsection{}
\label{subsec01}We study specific solutions to the following
quasilinear equation
\begin{equation}\label{main0}
\begin{split}
u_{xx}\left((\g+1)u^2_x+(\g-1)u_y^2\right)+4u_{xy}u_xu_y+
u_{yy}((\g+1) u^2_y+(\g-1)u_x^2)&=0
\end{split}
\end{equation}
where $ |\gamma|>1$ or $\gamma=1$. Let $L_\gamma[u]$ denote the left-hand side of this equation. A solution of the form
\begin{equation} u(x,y)=\rho ^kf(\theta), \quad k\geq 1, \label{resh}
\end{equation}
where $\rho$ and $\theta $ are the polar coordinates in the
$(x,y)$-plane, is said to be a \textit{quasiradial}. The origin of this study goes back to the well-known
$p$-Laplace equation
\begin{equation*}\label{equ:p-laplace}
\Div \left(|\nabla u|^{p -2}\nabla u\right)=0,
\end{equation*}
which is the divergence form (\ref{main0}) with $ p
=\frac{2\g}{\g-1}$. We call solutions to (\ref{main0})
$\gamma$-harmonic functions.

The existence and integral representations for quasiradial
$\gamma$-harmonic functions were previously established by G.~
Aronsson in \cite{Ar2}--\cite{Ar91}. In particular, it was shown in
\cite{Ar3} (see also \cite{Persen89} and \cite{Ar86}) that
\textit{single-valued} quasiradial $\gamma$-harmonic functions do
exist only for those exponents $k$ in (\ref{resh}) which satisfy the
characteristic equation
\begin{equation}\label{equ}
(2N-1)(\g+1)k^2-2(N^2\g+2N-1)k+N^2(1+\g)=0, \quad n\in\N{}.
\end{equation}
We refer  to the corresponding solution as the $N$-\textit{solution}
to (\ref{main0}). One can easily see that the $0$-solutions are
constants, and the $1$-solutions are linear functions. In what
follows, such solutions are said to be the \textit{trivial}
solutions to (\ref{main0}).

The limit case $\gamma=\infty$ reduces to the standard Laplace
equation $\Delta u=0$, and the corresponding $N$-solutions are
harmonic polynomials of degree $N$. Note that harmonic polynomials
are \textit{algebraic} functions, i.e. they satisfy some (actually,
trivial one) polynomial identity $P(x,y,u)\equiv 0$. We show that
this property is still valid for $N$-solutions of the Aronsson
equation (i.e. $\gamma=1$)
\begin{equation}
u_{xx}u^2_x+2u_{xy}u_xu_y+u_{yy}u^2_y=0, \label{equ:aron}
\end{equation}

It is the aim of this paper to study this phenomena for general
$\gamma$'s. To make this point more explicit, we have to note that
for $\g\ne\infty$ a (weak) solution of (\ref{main0}) is normally in
the class $C^{1,\alpha}$. In particular, quasiradial solutions have
a H\"older singularity near the origin, and one should consider them
as `singular solutions' (the terminology is borrowed from
\cite{Ar89}). This non-regular character is a consequence of the
general situation for $\gamma$-harmonic functions near their
singular points (i.e. at the points at which $|\nabla u|=0$), see
\cite{Ur}, \cite{Ev82}, \cite{Lew83}.

Our key result is the following explicit parametric representations
for $N$-solutions:
\begin{equation}\label{equ:poly}
\begin{split}
x+iy&=e^{i\phi}(\mu \zeta|\zeta|^{2(N-1)}+\bar{\zeta}^{2N-1}),\\
u_N&=C|\zeta|^{k(2N-1)-N}\cdot\re \zeta^N.
\end{split}
\end{equation}
Here $\zeta\in\Com{}$ is the parametrization variable, $\phi$ is an
arbitrary constant, $k=k(N,\g)$ is the biggest root of (\ref{equ}),
and $\mu$ is defined by (\ref{equ:mu-g-n}) below. We show also that
(\ref{equ:poly}) represents an entire graph over the $(x,y)$-plane
when $\zeta$ runs the  complex plane $\Com{}$.

In particular, it immediately follows  from (\ref{equ}) and
(\ref{equ:poly}), that $u_N$ is an \textit{algebraic} function
whenever $k(N,\gamma)$ is a \textit{rational} number. This makes
more explicit the mentioned above H\"olderian behaviour of
quasiradial solutions at their singular points.

In contrast with (\ref{equ:aron}), we show that for any rational
$|\gamma|>1$, the class of algebraic $N$-solutions is necessary
finite in the sense that the following upper estimate holds $ N\leq
\left\lfloor\frac{q^2(p^2+2-q^2)}{2p^2}\right\rfloor$, where
$\gamma=p/q$, with $q$ and $p$ to be co-prime numbers, and $\lfloor
x\rfloor$ denotes the integer part of $x$. In particular, this
yields the absence of algebraic solutions for all integer numbers
$\g$, $|\g|\geq 2$.

Denote by $\mathcal{A}$ the set of all $\gamma\in\Q{}$,
$|\gamma|>1$, such that (\ref{main0}) admits \textit{nontrivial}
algebraic $N$-solutions. Then $\gamma\in\mathcal{A}$ iff
$-\gamma\in\mathcal{A}$ (see Section~\ref{subsec:conj} below). In
Section~\ref{sec:algebraic} we show that for all rational
$\gamma=p/q\in\mathcal{A}$, $ \gamma>1$, the following bilateral
estimate holds
$$
 \quad 2+q\leq p\leq q^2-2.
$$

The last inequalities are sharp. Moreover, in Section~\ref{sub:max}
we prove that equality $q=p+2$ holds iff $q$ is an odd number,
$q\geq 3$. This yields the so-called minimal series
\begin{equation}\label{q-q-q}
\g=\frac{2N+3}{2N+1}, \qquad N=2,3,4,\ldots,
\end{equation}
and $N$ is the index of the corresponding (a unique for the given
$\gamma$) algebraic quasiradial solution.

We end this introduction with one possible physical interpretation
of (\ref{q-q-q}), which gives a motivation of our choice of $\gamma$
instead of $p$. Namely, observe that (\ref{main0}) is a homogeneous
form of the gas dynamics equation
\begin{equation}
L_\gamma[\phi]=2\Delta \phi, \label{main+}
\end{equation}
for the potential of the gas velocity  $\phi$  \cite[p.~9]{Bers}.
The parameter $\g$ in (\ref{main+}) is the so-called
\textit{adiabatic gas constant}, that is the ratio of the gas'
specific heats at constant volume and constant pressure (see, e.g.
\cite{Tom}). Note that, for all known gas models due to a specific
combinatoric nature of the adiabatic constant $\gamma$, it can be
evaluated in terms of the freedom degrees of the corresponding gas.
In particular, it follows that $\gamma$ a \textit{rational} number.

Moreover, the most important for applications is the simple gas
which consists of $m$ atoms, $\gamma$ is given by the following
ratio
\begin{equation}\label{equ:adia}
\gamma_{m}=\frac{2m+3}{2m+1}, \qquad m=1,2,\ldots,
\end{equation}
(e.g., $m=2$, or $\g_2=7/5$, describes the standard Earth's
atmosphere).

Note that eq. (\ref{main+}) is of non-degenerate elliptic type for
all adiabatic exponents $\gamma>1$. In the recent paper of I.~Zorina
and the author \cite{Zor}, it is proved that for any integer $N\geq
2$ there exists a solution of (\ref{main+}) with non-trivial
polynomial growth $k_N>1$ (where $k_N$ is defined by (\ref{equ}))
which is a real analytic function in the whole $\R{2}$.

Now, the quasiradial $N$-solutions to (\ref{main0}) can be naturally
regarded as the cones (after a suitable scale renormalization) over
the corresponding $N$-solutions to (\ref{main+}). In other words,
(\ref{main0}) represents a microscopy level of the gas flow. In this
connection, the coincidence of the minimal series (\ref{q-q-q}) and
the natural adiabatic constants (\ref{equ:adia}) implicitly
underlines the essence of algebraic character of the corresponding
$N$-solutions. We observe also, that in this case we have $N=m$,
i.e. the atomic number is equal to the index of the corresponding
$N$-solution.

\section{The separation equation}\label{subsec:22}
In this section we study the basic properties of the wave function
$f(\theta)$. For technical reasons, we assume that $|\gamma|>1$.
The case $\gamma=1$ requires a few more care because of degenerate
character of the separate equation (\ref{fte}). However,  all the
formulated below results are still valid in this limit case if we
suppose that $k>1$.

The separation of variables in (\ref{resh})  yields the following
ordinary differential equation
%\begin{equation}
\begin{multline}
f''((\g-1)k^2f^2+(\g+1)f'^2)+\\
+\biggl[f'^2(k(\g+3)-2)+((1+\g)k-2)k^2f^2\biggr]fk=0,
 \label{ft}
 \end{multline}
%\end{equation}
where the prime denotes the derivative with respect to  $\theta$.
Letting $W=f'^2(\theta)$, $Z=f^2(\theta)$, we can rewrite
(\ref{ft}) as
\begin{equation*}\label{WZ}
  \frac{dW}{dZ}=-\frac{k((\g+3)k-2)W+k^3((\g+1)k-2)Z}{(\g+1)W+(\g-1)k^2Z},
\end{equation*}
which splits into the following linear system
\begin{equation}
\label{WZ1}
    \begin{array}{lcl}
    W'(\xi)&=&-k((\g+3)k-2)W-k^3((\g+1)k-2)Z  \\
    Z'(\xi)&=& (\g+1)W+(\g-1)k^2Z.
    \end{array}
\end{equation}
One can easily verify that
\begin{equation*}
\label{WZ2}
    \begin{array}{lcl}
    W(\xi)&=&C_1k^2e^{-2k^2\xi}+C_2k((\g+1)k-2)e^{2(k-k^2)\xi}\\
    Z(\xi)&=&-C_1e^{-2k^2\xi}-C_2(\g+1)e^{2(k-k^2)\xi},
    \end{array}
\end{equation*}
is the general solution of (\ref{WZ1}), where $C_1$ and $C_2$ are
arbitrary constants. Then
\begin{equation}
\label{WZ3}
    W+k^2Z=-k\eta^{k-1}C_2\qquad
    W+\lambda ^2Z=\frac{2k}{\g+1}\eta ^kC_1,
\end{equation}
where $\eta=e^{-2\xi k}$ and $ \lambda ^2=k^2-\frac{2k}{\g+1}. $
The right-hand side of the latter identity is positive for all
$|\g|>1$ and $k\geq1$, so we can define
\begin{equation}\label{lambda-def}
\lambda:=\sqrt{k^2-\frac{2k}{\g+1}}.
\end{equation}
Then elimination of $\eta$ in (\ref{WZ3}) yields
$
%\begin{equation*}%\label{Cl}
(W+k^2Z)^k=C_3(W+\lambda ^2Z)^{k-1}.
%\end{equation*}
$
Since (\ref{ft}) is a homogeneous equation, it suffices to study
only  the case $C_3=1$. Thus, we have the following first order
differential equation
\begin{equation}
\label{fte} (f'^2(\theta)+k^2f^2(\theta))^k=(f'^2(\theta)+\lambda
^2f^2(\theta))^{k-1}.
\end{equation}
We introduce  new phase variables $ z=f(\theta)$, $w=f'(\theta)$,
and define the set
\begin{equation}
\label{faz} \Gamma=\{(z,w)\in\R{2}:\; (w^2+k^2z^2)^k=(w^2+\lambda
^2z^2)^{k-1}, \; w^2+z^2\ne 0\}.
\end{equation}
Observe  that the intersection $\Gamma$ with the $Oz$-axis
consists of exactly two points (the \textit{apexes}):
$A^{\pm}=(\pm z_0,0)$, where
\begin{equation}\label{z0}
z_0=\lambda^{k-1}k^{-k}.
\end{equation}

\begin{lem}\label{lem:graph}
Let $|\g|>1$ and $k>1$. Then $\Gamma$ a real analytic closed
Jordan curve. Moreover, $\Gamma\setminus\{A^+,A^-\}$ splits in two
mutually symmetric graphs (which have no common points with $Oz$).
\end{lem}

\begin{proof}
The first statement easily follows from the representation of
$\Gamma$ in the polar coordinates $z=r\cos\alpha$,
$w=r\sin\alpha$:
$$
r=\frac{(\sin^2\alpha+\lambda^2
\cos^2\alpha)^{(k-1)/2}}{(\sin^2\alpha+k^2\cos^2\alpha)^{k/2}}.
$$

Now, consider $\Gamma$ as the $0$-level set of the function
\begin{equation}\label{F}
F(z;w)=(w^2+k^2z^2)^k-(w^2+\lambda ^2z^2)^{k-1}.
\end{equation}
We claim that $F'_w\ne0$ on $\Gamma\setminus\{A^+, A^-\}$. Indeed,
suppose $F'_w(z_1,w_1)=0$ and $w_1\ne 0$. Then we have
\begin{equation}
\label{grF} \frac{1}{2w_1}F'_w
(z_1,w_1)=k(w_1^2+k^2z_1^2)^{k-1}-(k-1)(w_1^2+\lambda
^2z_1^2)^{k-2}=0,
\end{equation}
which together with $F(z_1,w_1)=0$ and $|\gamma|>1$ implies
$$
w_1^2=\frac{1-\g^2}{(1+\g)^2}k^2z_1^2\leq 0.
$$
It follows that  $w_1=z_1=0$, which contradicts the definition of
$\Gamma$ and implies our claim. Thus, $\Gamma\setminus\{A^+,
A^-\}$ splits into union of two graphs with respect to the
$Oz$-axis and the lemma is proved. \qed
\end{proof}

Now, we construct a special solution $f(\theta)$ of (\ref{fte}),
satisfying the initial condition $f(0)=z_0$. With the above
notation we have
$$
\frac{dz}{w}=\frac{f'(\theta)d\theta}{f'(\theta)}=d\theta.
$$
Define
\begin{equation}\label{e:int}
\Theta(\xi)=\int_{A^+}^\xi \frac{dz}{w},
\end{equation}
where $\xi\in\Gamma$, and the integral is taken clockwise along
the arc $(A^+,\xi)$ of $\Gamma$.  The last integrand a priori has
singular behavior when $w$ vanishes (i.e. for $\xi=A^{\pm}$). But,
it can be  shown that these singularities are removable. Indeed,
using again the representation of $\Gamma$ as the $0$-level set of
function (\ref{F}) we find
\begin{equation}\label{denom}
\frac{dz}{w}=-\frac{F'_wdw}{F'_zw}=-\frac{k(w^2+k^2z^2)^{k-1}-(k-1)(w^2+\lambda
^2z^2)^{k-2}}{k^3(w^2+k^2z^2)^{k-1}-(k-1)\lambda ^2(w^2+\lambda
^2z^2)^{k-2}}\cdot\frac{dw}{z}.
\end{equation}
Now, to show that (\ref{e:int}) has no singularity  it suffices
only to verify that the denominator of the right-hand ratio in
(\ref{denom}) is non-zero in a neighborhood of the apexes. The
corresponding values at $A^{\pm}$ are equal
$\lambda^{2(k-1)}z_0^{2(k-2)}\ne0$. Therefore, the integral in
(\ref{e:int}) is well defined, and it follows that $\Theta(\xi)$
is an analytic function of $\xi$ (in the sense that
$\Theta(\xi(\tau))$ is analytic for any analytic parametrization
$\xi(\tau)$).

Next, observe that $dz/w>0$ within our convention. Thus, applying
Lemma~\ref{lem:graph}, we conclude that $\Theta(\xi)$ is a
strictly increasing function when $\xi$ runs $\Gamma$ in clockwise
direction. Define function $f_k(\theta)$ by letting $
f_k(\Theta(\xi))=z(\xi), $ where $z(\xi)$ is the projection of
$\xi$ onto the $Oz$-axis. Clearly, $f_k(\theta)$ becomes a real
analytic periodic function, which is defined now in $\R{}$. By
virtue of the symmetry of $\Gamma$, the one-quoter period $T$
satisfies
\begin{equation}\label{equ:TTT} \frac{T}{4}:=\int_{A^+}^{B^{-}}
\frac{dz}{w}
\end{equation}
where $B^-=(-1,0)\in\Gamma$.

\begin{lem}\label{lem:per}
\begin{equation}\label{equ:period1}
T=2\pi\left(1-\frac{k-1}{\lambda }\right).
\end{equation}
\end{lem}

\begin{proof}
Define a new variable $t$ by letting $w=-zt$. Then, applying
(\ref{faz}) we obtain the following parameterization of the arc
$A^{+}B^{-}$
\begin{equation}\label{equ:z-w-t}
\begin{split}
z(t)=&(t^2+\lambda^2)^{(k-1)/2}(t^2+k^2)^{-k/2}\\
w(t)=&-t(t^2+\lambda^2)^{(k-1)/2}(t^2+k^2)^{-k/2},
\end{split}
\end{equation}
When $t$ runs between $0$ and $+\infty$ the corresponding point
$\xi(t)=(z(t),w(t))$ runs $A^{+}B^{-}$ in the clockwise direction.
Moreover, we have
\begin{equation*}%\label{equ:new_t}
\begin{split}
dz(t)&=t(t^2+\lambda^2)^{\frac{k-3}{2}}(t^2+k^2)^{-\frac{k+2}{2}}\biggl(k^2(k-1)-k\lambda^2-t^2\biggr)\;dt,\\
%&=-t(t^2+\lambda^2)^{\frac{k-3}{2}}(t^2+k^2)^{-\frac{k+2}{2}}\left(\frac{\g-1}{\g+1}k^2+t^2\right)\;dt,
\end{split}
\end{equation*}
which yields
\begin{equation}\label{equ:Ostrogr}
\frac{dz(t)}{w(t)}
=\left(\frac{k}{t^2+k^2}-\frac{k-1}{t^2+\lambda^2}\right)dt.
\end{equation}
Integration of (\ref{equ:Ostrogr}) gives
\begin{equation}\label{equ:theta_def}
\Theta(z(t),w(t))=\arctan\frac{t}{k}-\frac{k-1}{\lambda }\arctan
\frac{t}{\lambda }.
\end{equation}
So, letting $t\to+\infty$ in (\ref{equ:theta_def}), we obtain by
virtue of (\ref{equ:TTT}): $
\frac{T}{4}=\frac{\pi}{2}\left(1-\frac{k-1}{\lambda }\right) $, and
(\ref{equ:period1}) follows. \qed \end{proof}

From the definition of $\Gamma$ we infer $
f_k'(\Theta(\xi))=w(\xi), $ where $z(\xi)$ is the projection of
$\xi$ onto the $Oz$-axis. In particular,
\begin{equation}\label{null}
f_k'(\theta)=0 \quad\Leftrightarrow \quad
\theta=\frac{Tn}{2},\quad n\in\Z{}.
\end{equation}
 Hence, the constructed function $f_k(\theta)$
satisfies (\ref{fte}) with the initial data
\begin{equation*}\label{equ:f'0}
f_k'(0)=0, \quad f_k(0)=z_0\equiv\frac{\lambda ^{k-1}}{k^{k}}.
\end{equation*}
It is obvious that (by virtue of the autonomic character of
(\ref{fte})) a general solution of (\ref{fte}) has the form $
f(\theta)=Cf_k(\theta+a), $ where $C$ and $a$ are arbitrary
constants.

\begin{cor}
Let $T$ be defined by (\ref{equ:period1}). For
$\theta\in(-T/4,T/4)$ we have the following parametrization for
$f_k(\theta)$
\begin{equation}\label{new-par1}
\begin{split}
f_k=&(t^2+\lambda^2)^{(k-1)/2}(t^2+k^2)^{-k/2}\\
\theta=&\arctan\frac{t}{k}-\frac{k-1}{\lambda }\arctan
\frac{t}{\lambda },\quad t\in\R{},
\end{split}
\end{equation}
and for other values of $\theta$, $f(\theta)$ satisfies the
symmetry rules:
$$
f_k(\theta)=-f_k(\frac{T}{2}-\theta), \qquad
f_k(-\theta)=f_k(\theta).
$$
\end{cor}

Our further objective is to characterize all values of $k>1$ which
support the $2\pi$-periodic wave functions $f_k(\theta)$, that is
$ f_k(\theta+2\pi)=f_k(\theta). $ In fact, the latter condition is
equivalent to absence of multivalued branches of the corresponding
quasiradial solution of (\ref{main0}). The following assertion is
a direct consequence of Lemma~\ref{lem:per}.

\begin{prop}
\label{corol:period} Let $|\g|>1$ and $k\geq 1$. Then
$f_k(\theta+2\pi)=f_k(\theta)$
  if and only if the
following equality holds
\begin{equation}\label{equ:N-def}
\frac{k-1}{\lambda }=\frac{N-1}{N},\qquad N\in\N{},
\end{equation}
where $\N{}$ denotes the set of all positive integers. In this
case $f_k(\theta)$ is a $2\pi/N$-periodic function.
\end{prop}

The latter statement is not new. It has appeared in \cite{Ar3} for
$\g=1$ and  for general $\gamma$ in \cite{Ar86}, \cite{Persen89}.
However, our approach to this result seems to be complementary to
these previous works and allows us to arrive at an explicit
representation for the quasiradial solutions in the next section.

\begin{prop}\label{corol:11}
Let $|\g|>1$ and $N\in\N{}$ be given. Then there exists a unique
$k=k(\g,N)\geq1$ such that (\ref{ft}) admits a $2\pi/N$-periodic
solution.
\end{prop}

\begin{proof}
The trivial case $N=1$ by (\ref{equ:N-def}) gives $k=1$. Let
$N\geq 2$. Then (\ref{equ:N-def}) is equivalent to the following
quadratic equation
\begin{equation}\label{k-l}
(2N-1)(\g+1)k^2-2(N^2\g+2N-1)k+N^2(1+\g)=0
\end{equation}
which has two separate roots because it discriminant is strictly
positive:
$$
D=4(N-1)^2(N^2\g^2-2N+1)>4(N-1)^3>0.
$$
Then one can easily infer from the Vi\'ete theorem that
\begin{equation}\label{last}
(k_1-1)(k_2-1)=-\frac{(N-1)^2}{2N-1}\cdot\frac{\g-1}{\g+1}<0,
\end{equation}
where $k_1\ne k_2$ are the roots of (\ref{k-l}). Inequality
(\ref{last}) implies $k_1<1<k_2$, so that exactly one root $k_2>1$
is consistent with our  constraint $k>1$. \qed
\end{proof}

\section{$N$-solutions}

From now on, we adopt a new notation $f_N$ for the $N$-th wave
function $f_k$ with $k=k(\gamma,N)$.

\begin{defn}
Let $N\in\N{}$. The quasiradial solution of the form
$$
u_N(x,y):=C\rho^k f_N(\theta),
$$
where $C$ is an arbitrary constant, is said to be a \textit{basic}
$N$-solution of (\ref{main0}). Similarly, $u=C\rho^k
f_N(\theta+a)$ with an arbitrary $a\in\R{}$ is said to be a
(general) $N$-solution.
\end{defn}

\begin{thm}\label{lem:repres}
Let $|\gamma|>1$, $N\in\N{}$, and $k=k(\gamma,N)$ be the biggest
root of (\ref{k-l}). Then the basic $N$-solution has the following
representation
\begin{equation*}\label{equ:resh1}
\begin{split}
x&=h^{2N-1} ((k+\lambda)\cos \tau+(k-\lambda)\cos (2N-1)\tau),\\
y&=h^{2N-1} ((k+\lambda)\sin \tau-(k-\lambda)\sin (2N-1)\tau),\\
u_N&=Ch^{ k(2N-1)}\cos N\tau,
\end{split}
\end{equation*}
where $\lambda$ is defined by (\ref{equ:N-def}), and
$\tau\in[0;2\pi]$, $h>0$ are the variables of parametrization.
\end{thm}

\begin{proof}
By virtue of (\ref{new-par1}) we have the following
parametrization for $f_N(\theta)$
\begin{equation}\label{new-par}
\begin{split}
f_N=&(t^2+\lambda^2)^{(k-1)/2}(t^2+k^2)^{-k/2}\\
\theta=&\arctan\frac{t}{k}-\frac{k-1}{\lambda }\arctan
\frac{t}{\lambda }.
\end{split}
\end{equation}
Define a new variable $\tau$ by
\begin{equation}\label{equ:N-tau}
t=\lambda \tan (N\tau), \quad
\tau\in(-\frac{\pi}{2N},\frac{\pi}{2N}).
\end{equation}
Then we have from (\ref{new-par}) and (\ref{equ:N-def})
\begin{equation}\label{equ:def-tau}
\theta+(N-1)\tau=\arctan\left(\frac{\lambda}{k}\tan
(N\tau)\right),
\end{equation}
which by (\ref{equ:N-tau}) yields
\begin{equation}\label{equ:N-t}
t=k\tan \left(\theta+(N-1)\tau\right).
\end{equation}
Inserting (\ref{equ:N-tau}) and (\ref{equ:N-t}) into the first
identity in (\ref{new-par}) we get
\begin{equation}\label{equ:f-N}
\begin{split}
f_N=&
(t^2+\lambda^2)^{\frac{k-1}{2}}(t^2+k^2)^{-\frac{k}{2}}=\\
=&z_0\left(1+\tan^2 (N\tau)\right)^{\frac{k-1}{2}}\,\left(1+ \tan
^2
(\theta+(N-1)\tau)\right)^{-\frac{k}{2}}=\\
=&z_0\cos N\tau\;\left(\frac{\cos(\theta+(N-1)\tau)}{\cos
N\tau}\right)^k,
\end{split}
\end{equation}
where $z_0$ is defined by (\ref{z0}). On the other hand,
\begin{equation}\label{equ:cos}
\cos(\theta+(N-1)\tau)=\cos\theta\cos(N-1)\tau\cdot[ 1-\tan
(N-1)\tau \cdot\tan \theta].
\end{equation}
Now, we apply to (\ref{equ:cos})  the addition formula
$$
1-\tan p  \cdot \tan(\beta-p )=\frac{1}{\cos^2p }\cdot
\frac{1}{1+\tan\beta\tan p }
$$
with $p =(N-1)\tau$, $\beta=\arctan\left(\frac{\lambda}{k}\tan
(N\tau)\right)$. By virtue of (\ref{equ:def-tau}), we have
$\theta=\beta-p$, therefore
\begin{equation}\label{equ:tan}
\begin{split}
1-&\tan (N-1)\tau \tan \theta=\frac{1}{\cos^2(N-1)\tau}\cdot
\frac{1}{1+\frac{\lambda}{k}\tan(N-1)\tau\tan N\tau}=\\
&=\frac{1}{\cos (N-1)\tau}\cdot\frac{k\cos N\tau}{k\cos
(N-1)\tau\cos
N\tau+\lambda \sin (N-1)\tau\sin N\tau}=\\
&=\frac{1}{\cos (N-1)\tau}\cdot\frac{2k\cos N\tau}{(k+\lambda)\cos
\tau+(k-\lambda)\cos (2N-1)\tau}.
\end{split}
\end{equation}
Then, applying (\ref{equ:cos}) and  (\ref{equ:tan}) to
(\ref{equ:f-N}) we obtain
\begin{equation*}%\label{equ:f-theta}
f_N=z_0\cos N\tau\;\left[ \frac{2k\cos\theta} {(k+\lambda)\cos
\tau+(k-\lambda)\cos (2N-1)\tau}\right]^k.
\end{equation*}
Taking into account that $u_N(x,y)=\rho^kf_N$ and
$x=\rho\cos\theta$ we find
\begin{equation*}\label{equ:resh0}
\begin{split}
u_N(x,y)&=z_0\cos N\tau\;\left[ \frac{2k\cos\theta}
{(k+\lambda)\cos
\tau+(k-\lambda)\cos (2N-1)\tau}\right]^k\rho^k=\\
&=(2k)^kz_0 \left[\frac{x} {(k+\lambda)\cos \tau+(k-\lambda)\cos
(2N-1)\tau}\right]^k
\end{split}
\end{equation*}
Setting $h^{2N-1}$ for the expression in the last brackets we
arrive at
$$
x=h^{2N-1} ((k+\lambda)\cos \tau+(k-\lambda)\cos (2N-1)\tau),
$$
and $ u_N=C_Nh^{k(2N-1)}\cos N\tau $, where
$C_N=(2k)^kz_0=2^k\lambda^{k-1}$.

Finally, to express $y$ we eliminate the polar coordinates as
follows
\begin{equation*}
\begin{split}
\frac{y}{x}=\tan \theta&=\frac{\lambda\tan N\tau-k\tan
(N-1)\tau}{k+\lambda \tan N\tau\tan (N-1)\tau }=\\
&=\frac{(k+\lambda)\sin \tau-(k-\lambda)\sin
(2N-1)\tau}{(k+\lambda)\cos \tau+(k-\lambda)\cos (2N-1)\tau}.
\end{split}
\end{equation*}
Thus we get (\ref{equ:resh1}) for all $\tau\in(-\pi/2N,\pi/2N)$.
Using the analyticity of $f_N(\theta)$ we conclude that
(\ref{equ:resh1}) is valid for all $\tau$. The theorem is proved
completely. \qed
\end{proof}

In order to simplify (\ref{equ:resh1}) we make use a special
intermediate parameter $\mu$. Namely, in the above notation, put
\begin{equation}\label{e:34}
\mu=\frac{k+\lambda}{k-\lambda},
\end{equation}
so $\lambda=k(\mu-1)/(\mu+1)$, and we have from (\ref{lambda-def})
\begin{equation}\label{equ:v}
k=\frac{(1+\mu)^2}{2\mu(\g+1)}.
\end{equation}
Then, an easy computation shows that (\ref{equ:N-def}) becomes
\begin{equation}\label{equ:N-v}
N=\frac{\mu^2-1}{2(\mu\g-1)}.
\end{equation}
Now we observe that for $\gamma>1$ we have $k>\lambda>0$, so that
$\mu>1$. Similarly, $\gamma<-1$ implies $\mu<-1$. On the other
hand, considering (\ref{equ:N-v}) as a quadratic equation for
$\mu$
\begin{equation}\label{mu}
F(\mu):=\mu^2-2\gamma N\mu+(2N-1)=0,
\end{equation}
we find $F(1)=2N(1-\gamma)$ and $F(-1)=2N(1+\gamma)$. If
$\gamma>1$ then $F(1)<0$, i.e. exactly one root of (\ref{mu})
$$
\mu^+=N\g+\sqrt{N^2\g^2-2N+1}
$$
agrees the above constrain $\mu>1$. Similarly, for $\gamma<-1$ we
have
$$
\mu^-=N\g-\sqrt{N^2\g^2-2N+1}.
$$
Define
\begin{equation}\label{equ:mu-g-n}
\mu\equiv \mu(\g,N)=\left\{%
\begin{array}{ll}
    N\g+\sqrt{\vphantom{\int}N^2\g^2-2N+1}, & \hbox{ for $\g>1$;} \\
    N\g-\sqrt{\vphantom{\int}N^2\g^2-2N+1}, & \hbox{ for $\g<-1$.} \\
\end{array}%
\right.
\end{equation}
For the further purpose, notice that
\begin{equation}\label{equ:xi-xi}
\mu(-\g,N)=-\mu(\g,N).
\end{equation}
Now, by using the homogeneity character of (\ref{equ:resh2}), one
can rewrite it as follows
\begin{equation}\label{equ:resh2}
\begin{split}
x&=h^{2N-1} (\mu\cos \tau+\cos (2N-1)\tau),\\
y&=h^{2N-1} (\mu\sin \tau-\sin (2N-1)\tau),\\
u_N&=Ch^{(2N-1)k}\cos N\tau.
\end{split}
\end{equation}

A general $N$-solution can be obtained from a certain basic
$N$-solution by a suitable rotation of  in the $(x,y)$ plane. Let
$$
x'=x\cos\psi+y\sin\psi, \qquad y'=-x\sin\psi+y\cos\psi,
$$
be such a rotation. Then (\ref{equ:resh2}) implies
\begin{equation}\label{equ:resh4}
\begin{split}
x&=h^{2N-1} (\mu\cos \tau+\cos ((2N-1)\tau+2N\psi)),\\
y&=h^{2N-1} (\mu\sin \tau-\sin ((2N-1)\tau+2N\psi)),\\
u&\equiv u_{N,\psi}=Ch^{(2N-1)k}\cos N(\tau+\psi),
\end{split}
\end{equation}
In particular, $u_{N,0}=u_N$. Thus, (\ref{equ:resh4}) gives the
representation for general $N$-\textit{solutions} to
(\ref{main0}).

By putting in (\ref{equ:resh4}) $X=h\cos\tau$ and $Y=h\sin\tau$,
we obtain an \textit{algebraic} representation of a basic
$N$-solution $u(x,y)$ (see also an equivalent complex form
(\ref{equ:poly}) given in the Introduction)
\begin{equation}\label{equ:reshPHI}
\begin{split}
x&=\mu X(X^2+Y^2)^{N-1}+\re (X+iY)^{2N-1},\\
y&=\mu Y(X^2+Y^2)^{N-1}-\im (X+iY)^{2N-1},\\
 u_N&=C(X^2+Y^2)^{\frac{k(2N-1)-N}{2}}\re (X+iY)^{N}.
\end{split}
\end{equation}

\begin{cor}
All $N$-solutions are quasialgebraic functions in the sense that
$u^{\alpha}_N$ is an algebraic function, where
$\alpha=\frac{2}{k(2N-1)-N}$.
\end{cor}

To illustrate the last property, we briefly mention the following
well-known example. Note (see, also \cite{Ar3} and \cite{JLM} for
further examples) that $ u_2=x^{4/3}-y^{4/3} $ is a basic
$2$-solution of (\ref{equ:aron}). Then it is easily verified that
$u\equiv u_2(x,y)$ satisfies the following polynomial identity
$$
27x^4y^4u^3=(x^4-y^4-u^3)^3.
$$

%We observe that the last parametrization is regular in the sense
%that it provides a polynomial homeomorphism
%$$
%\phi(X+iY):=x(X,Y)+iy(X,Y):\Com{}\to\Com{}.
%$$
%In fact, by homogeneity of $\phi$ it suffices to show that the
%restriction of $\phi$ on the unit circle
%\begin{equation}\label{equ:zeta}
%\Phi(z ):=\phi_1(z)+i\phi_2(z)=\mu z +\frac{1}{z ^{2N-1}}, \qquad
%z =e^{i\tau}
%\end{equation}
%is an injective embedding in the punctured plane
%$\Com{}\setminus\{0\}$. Assuming the converse, suppose that
%$\Phi(z )$ is not injective, i.e. there exists a pair of two
%distinct points $|z _1|=|z _2|=1$ on the unit circle and such that
%$\Phi(z _1)=\Phi(z _2)$. Then we find from (\ref{equ:zeta})
%$$
%\mu z _1^{2N-1}z _2^{2N-1}=\sum_{j=0}^{2N-2}z _1^jz _2^{2N-2-j}.
%$$
%In particular,
%\begin{equation*}\label{equ:1}
%|\mu|\leq \sum_{j=0}^{2N-2}|z _1|^j|z _2|^{2N-2-j}=2N-1.
%\end{equation*}
%The last inequality contradicts (\ref{inequal}). Therefore, the
%injectivity of $\Phi$ follows.
%
%Now, let $z $ be such that $\Phi(z )=0$. It follows then from
%(\ref{equ:zeta}) that $z ^{N}=-1/\mu$. Since $|z |=1$, then
%$|\mu|=1$, and (\ref{inequal}) yields  $(2N-1)<1$. But the latter
%contradicts the assumption $N\geq 2$. Thus, the injectivity of
%$\phi$ is proved.
%
%Since $\phi$ is homogeneous and $\phi(X,Y)\ne0$ for all $(X,Y)\ne
%0$ we conclude that $\phi$ is surjective.
%
%

\section{Conjugate solutions} \label{subsec:conj} We recall that
the main equation (\ref{main0}) can be represented as a
$p$-Laplace equation for $ p=\frac{2\g}{\g-1}.$ It is well known
fact (cf. \cite{Ar91}, \cite{ArLin}) that in two-dimensional case
there is the canonical correspondence between $p$-harmonic and
$p'$-harmonic functions for
$$
\frac{1}{p}+\frac{1}{p'}=1.
$$
More precisely, given  a solution $u$ of equation
$L_{\gamma}[u]=0$ we define the \textit{conjugate} function $U$ by
\begin{equation}\label{equ:U-u}
U_x=|\nabla u|^{2/(\g-1)}u_y, \qquad U_y=-|\nabla
u|^{2/(\g-1)}u_x, \qquad U(0)=0.
\end{equation}
Note that  $U(x,y)$ is not necessarily a single valued function.
But at least locally,  $U$ is a quasiradial solution of the
\textit{conjugated} equation
\begin{equation}
L_{-\gamma}[U]=0 \label{mainU}.
\end{equation}

It turns out that there is a simple relation between the conjugate
\textit{quasiradial} solutions. We define an \textit{adjoint} (to
basic one) $N$-solution $u_N^*$ as follows
\begin{equation}\label{equ:resh_adj}
\begin{split}
x&=h^{2N-1} (\mu\cos \tau-\cos ((2N-1)\tau)),\\
y&=h^{2N-1} (\mu\sin \tau+\sin ((2N-1)\tau)),\\
u^*_{N}&=Ch^{(2N-1)k}\sin N\tau.
\end{split}
\end{equation}
Applying (\ref{equ:resh4}), we obtain an equivalent definition:
$$
u^*_{N}=u_{N,-\pi/2N}.
$$
We have also $ u^{**}_{N}=-u_{N}$. The functions $u_N$ and $u^*_N$
form a \textit{conjugate} pair, analogous to conjugate harmonic
functions. More precisely, if $\gamma=1$ one can easily derive
from the above representation that $u_N(x,y)=\re(x+iy)^N$ and
$u^*_N(x,y)=\im(x+iy)^N$.

\begin{thm}\label{theo:conjugate}
Let $u_N$ be a basic $N$-solution of (\ref{main0}) and $U^*_N$ be
an adjoint $N$-solution of (\ref{mainU}). Then there is a constant
$c$ such that $cu_N$ and $U^*_N$ form the conjugate pair in the
sense (\ref{equ:U-u}).
\end{thm}

\begin{proof}
Without loss of generality, we can assume $N\geq2$. Let $u_N$ be
an $N$-solution of (\ref{main0}) and $U$ be the corresponding
conjugate function defined by (\ref{equ:U-u}), normalized by
$U(0)=0$. It follows from homogeneity of $u_N$ that $U$ is a
homogeneous function as well. Hence, there is a real $\beta$ such
that
\begin{equation}\label{G}
U=r^{\beta}G(\theta).
\end{equation}
Here $G(\theta)$ is a priori a multivalued function. We will prove
that $U=U^*_N$ with a suitable constant $C$ in
(\ref{equ:resh_adj}).

First, notice that $U$ is a quasiradial solution of (\ref{mainU}).
Denote by $k$ the growth exponent of $u_N$. Then $k>1$ and the
components of the gradient $\nabla u_N$ are homogeneous functions
of order $(k-1)$. Moreover,
\begin{equation*}\label{equ:grad-rho}
|\nabla u_N|^2= r^{2k-2}(k^2f_N^2(\theta)+f_N'^2(\theta)).
\end{equation*}
In particular, $\nabla u_N\ne0$ for  $r\ne 0$. Applying
(\ref{equ:U-u}) gives
$$
\nabla U=\begin{bmatrix}
  U'_x \\
  U'_y\\
\end{bmatrix}=
|\nabla u_N|^{2/(\g-1)}
\begin{bmatrix}
  u'_{N,y} \\
  -u'_{N,x} \\
\end{bmatrix}=
r^{(k-1)\frac{\g+1}{\g-1}}\begin{bmatrix}
   G_1(\theta)\\
   G_2(\theta)
   \end{bmatrix},
$$
where $G_i(\theta)$ are certain $2\pi$-periodic functions of
$\theta$. Since $|\g|>1$ and $k>1$,
$$
\beta =(k-1)\frac{\g+1}{\g-1}+1>1.
$$

On the other hand, let $k^*$ be the growth exponent of  $U^*_N$.
Let $\mu$ and $\mu^*$ be the corresponding auxiliary parameters
defined by (\ref{e:34}) for $k$ and $k^*$ respectively. Then it
follows that from (\ref{equ:xi-xi}) that $\mu^*=-\mu$, and from
(\ref{equ:v}) we have
$$
k=\frac{(1+\mu)^2}{2(1+\g)\mu}
$$
and
$$
k^*=\frac{(1+\mu^*)^2}{2(1+\g^*)\mu^*}=\frac{(1-\mu)^2}{2(\g-1)\mu},
$$
where $\gamma^*=-\gamma$. Hence
\begin{equation}\label{equ:k-+}
k(1+\g)+k^*(1-\g)\equiv k(1+\g)+k^*(1+\g^*)=2.
\end{equation}
Thus, by (\ref{equ:k-+}) we have
\begin{equation*}
\begin{split}
\beta&=(k-1)\frac{\g+1}{\g-1}+1=\frac{k(\g+1)}{\g-1}-\frac{\g+1}{\g-1}+1=\\
&=\frac{2-k^*(-\g+1)}{\g-1}-\frac{2}{\g-1}=k^*.
\end{split}
\end{equation*}
The latter means that $U$ is an $N$-solution of (\ref{mainU}) with
$ \beta=k^*$ (in particular, the function $G$ in (\ref{G}) is a
single valued function).

It remains only to show that $U$ is an \textit{adjoint}
$N$-solution. Since $U$ is an $N$-solution, there is $\psi$ such
that  $U=U_{N,\psi}$ in representation (\ref{equ:resh4}). Thus, we
have to prove that $N\psi=\pm\pi/2 \mod \pi$.

Choosing $\tau_0=-\psi+\pi/2N$, $h=1$ in representation
(\ref{equ:resh4}) for $U=U_{N,\psi}$, we obtain the corresponding
Cartesian coordinates $(x_0,y_0)\ne 0$. Then $U(x_0,y_0)=0$ and
applying the Euler theorem on homogeneous functions yields
\begin{equation}\label{coll}
 x_0\cdot {U'}_{x}(x_0,y_0)+
y_0\cdot\frac{\partial U}{\partial y}(x_0,y_0)=k^*U(x_0,y_0)=0.
\end{equation}
Hence, the gradient  $\nabla U(x_0,y_0)$ is orthogonal to the
radius vector $\nabla r(x_0,y_0)$ of the point $(x_0,y_0)$. On the
other hand, one can readily see from (\ref{equ:resh4}) that
\begin{equation*}
\begin{split}
x_0&=(\mu^*-1)\cos\left(\frac{\pi}{2N}-\psi\right),\\
y_0&=(\mu^*-1)\sin\left(\frac{\pi}{2N}-\psi\right),\\
\end{split}
\end{equation*}
hence we have for the polar angle of the point $(x_0,y_0)$
\begin{equation}\label{equ:compar}
\theta_0 \equiv \frac{\pi}{2N}-\psi \mod \pi.
\end{equation}
By (\ref{equ:U-u}), the gradients of $u_N$ and $U$ are mutually
orthogonal. Then from (\ref{coll}) we infer that  vectors $\nabla
u_N(x_0,y_0)$ and $\nabla r(x_0,y_0)$ are collinear. Since
$$
\nabla u_N=f_N(\theta)r^{k-1}\nabla r+r^{k}f_N'(\theta)\nabla
\theta,
$$
and
$$
\scal{\nabla\theta}{\nabla r}=0,
$$
we conclude that $f_N'(\theta_0)=0$. Hence, by (\ref{null}) there
is $n\in\Z{}$ such that
$$
\theta_0=\frac{Tn}{2}=\frac{\pi n}{N},
$$
which by (\ref{equ:compar}) yields $ N\psi\equiv \frac{\pi}{2} \mod
\pi, $ and the theorem follows.   \qed \end{proof}

\section{Algebraic $N$-solutions}
\label{sec:algebraic}

%\subsection{Preliminaries}\label{subsec:42}

In this section we settle the following question: \textsl{For
which rational numbers $\g\in\Q{}$ such that $|\g|>1$, equation
(\ref{main0}) does admit nontrivial (i.e. $N\geq 2$) algebraic
solutions?} Clearly, algebraicity of a general $N$-solution is
equivalent to that property for  basic $N$-solutions. Therefore,
in what follows we cofine ourselves  by only basic $N$-solutions.

\begin{lem}
\label{lem:raz} Let $\g\in \Q$, $|\g|>1$ and $N\in\N{}$. Then
inclusions $k(\gamma,N)\in \Q$, $\lambda(\gamma,N)\in \Q$ and
$\mu(\gamma,N)\in \Q$ are pairwise equivalent.
\end{lem}

\begin{proof}
It immediately follows from (\ref{equ:N-def}) that inclusions
$k\in \Q$ and $\lambda\in \Q$ are equivalent. By (\ref{equ:v}),
$\mu\in \Q$ implies both $k\in\Q{}$ and $\lambda\in\Q{}$. On the
other hand, if $k\in \Q$ then $\lambda\in \Q$, so by (\ref{e:34})
 $\mu\in \Q$ and the lemma is proved.   \qed \end{proof}

For the further convenience, we put $\g\in \mathcal{A}$,  if there
exists $N\geq 2$, $N\in\N{}$, such that $u_N(x,y)$ is an algebraic
function. In this case we also define
$$
\mathcal{N}({\g})=\{N\in\N{}: \text{$u_N$ is an algebraic
function}\}.
$$

\begin{lem}
\label{lem:algebra} Let $\g\in \mathcal{A}$, $|\g|>1$ and $N\geq
2$. Then $N\in \mathcal{N}({\g})$ iff the corresponding exponent
$k(\gamma,N)\in \Q$.
\end{lem}

\begin{proof}
Let $N\in \mathcal{N}({\g})$. Then it follows from
(\ref{equ:resh4}) and $|\mu|>2N-1\geq 3$ that
$$
(x^2+y^2)=h^{4N-2}[\mu^2+1+2\mu\cos (2N\tau+2\phi)]\sim h^{4N-2},
\qquad h\to\infty
$$
while
$$
u_n(x,y)=h^{(2N-1)k}\cos (2N\tau+\phi),
$$
where $k=k(\gamma,N)$. Thus, the growth exponent of $u$ is equal to
$k$ and it follows that $k\in\Q{}$. Now, let $k(\gamma,N)\in \Q$.
Then  (\ref{equ:reshPHI}) gives a rational parametrization of
$u^{2d}_N$, where $d$ is the denominator of $k$. Hence, $u_N(x,y)$
is an algebraic function. \qed \end{proof}

\begin{cor}
\label{cor:aron} For $\gamma=1$ all the $N$-solutions are
algebraic functions, i.e.  $\mathcal{N}(1)=\N{}$.
\end{cor}

The  next assertion  is an easy corollary of (\ref{equ:k-+})

\begin{lem}\label{lem:equiv_gamma}
Let $\g\in\Q{}$. Then $k_N\in\Q{} \Leftrightarrow k^*_N\in\Q{}$.
In particular,
$$
\mathcal{N}(\g)=\mathcal{N}(-\g).
$$
\end{lem}

By Lemma~\ref{lem:equiv_gamma}, we can assume without loss of
generality that $\g>1$. In what follows, we suppose that $p$ and
$q$ have no common divisors. Then by Lemma~\ref{lem:algebra},
$\g\in \mathcal{A}$ if and only if there is an integer $N\geq 2$
such that $k(\gamma,N)\in\Q{}$. By virtue of Lemma~\ref{lem:raz},
this is equivalent to existence of a rational solution
$\mu=\frac{A}{B}>1$ of (\ref{equ:N-v}):
$$
N=\frac{q(A^2-B^2)}{2B(Ap-Bq)}, \qquad A>B.
$$
Thus, we arrive at the following diophantine equation
\begin{equation}\label{equ:diaf}
A^2q-2ABpN+q(2N-1)B^2=0.
\end{equation}

\begin{thm}\label{theo:Q-P}
The following assertions are equivalent

(i) $\g=p/q\in \mathcal{A}$ with $N\in \mathcal{N}(\g)$, $N\geq
2$;

(ii)  equation (\ref{equ:diaf}) has  an integer solution $(A,B)$:
$A,B\in\Z{}$, and $A>B\geq1$;

(iii) the discriminant
\begin{equation}\label{equ:discr}
D=N^2p^2-q^2(2N-1)
\end{equation}
is an squared integer.

Moreover, if $\g\in \mathcal{A}$, $\gamma\ne1$, then the set
$\mathcal{N}(\g)$ is finite and the following upper bounds holds
\begin{equation}\label{equ:finite}
    N<\frac{q^2(p^2+2-q^2)}{2p^2}.
\end{equation}
In particular,
\begin{equation}\label{e-i}
 q\geq 3.
 \end{equation}
and
\begin{equation}\label{e-ii}
q+1\leq p \leq q^2-1;
\end{equation}
\end{thm}

\begin{proof}
Clearly, we have only to establish the equivalence (ii) and (iii).
In turn, the only nontrivial implication is (iii) $\Rightarrow$
(ii).

Let (iii) be true. Then $p=q\g>q$  and $D=d^2$ for some
$d\in\N{}$.Since
$$
D=N^2(p^2-q^2)+q^2(N-1)^2>0,
$$
we have $d>0$. Let $V=A/B$ and consider the following associated
with (\ref{equ:diaf}) quadratic equation
\begin{equation}\label{equ:diaf-V}
F(V):= V^2q-2pNV+q(2N-1)=0.
\end{equation}
Then (\ref{equ:diaf-V}) has two distinct \textit{rational}
solutions $v_1$ and $v_2$, $v_1<v_2$.  Since
$$
F(2N-1)=-2(N-1)(2N-1)(p-q)<0,
$$
we have $v_2>2N-1$. Let $v_2=A/B$ where $A$ and $B$ have no common
divisors, and $B>0$. Hence, $A>(2N-1)B>B$ and $(A,B)$ is a desired
solution of (ii).
%
%On the other hand, $F(1)=-2N(p-q)<0$ whence $v_1<1$ and it follows
%that the found above solution is unique (assuming that$A$ and $B$
%have no common divisors).

In order to establish the finiteness of $\mathcal{N}(\g)$, we make
use (iii). We have
\begin{equation}\label{equ:finite1}
 D=\left(Np-\frac{q^2}{p}\right)^2+\frac{q^2(p^2-q^2)}{p^2}.
\end{equation}
Let $D=d^2$, with $d\in\N{}$. Then (\ref{equ:finite1}) yields
\begin{equation}\label{equ:lower}
d>Np-\frac{q^2}{p},
\end{equation}
while the Bernoulli inequality and (\ref{equ:finite1}) imply the
upper following bound
\begin{equation}\label{equ:upper}
d<\left(Np-\frac{q^2}{p}\right)+\frac{q^2(p^2-q^2)}{2p(Np^2-q^2)}.
\end{equation}

Since $p$ and $q$ have no common divisors, we can write
$$
\frac{q^2}{p}=M+\frac{m}{p}
$$
where $M>0$ is an integer, and $m\in\{1,2,\ldots,p-1\}$ is the
non-zero remainder of the latter ratio. On the other hand, since
$m=q^2-Mp$, it follows that $m$ and $p$ have no common divisors.
Thus, by virtue of $Np\in\N{}$, and strict inequalities
(\ref{equ:lower}) and (\ref{equ:upper}) we obtain
$$
\frac{q^2(p^2-q^2)}{2p(Np^2-q^2)}>\frac{1}{p},
$$
which easily implies (\ref{equ:finite}).

To verify (\ref{e-i}) we notice that the cases $q=1$ and $q=2$
together with (\ref{equ:finite}) easily yield the contradiction:
$N<2$.

The left inequality in (\ref{e-ii}) immediately follows from
$p>q$. Finally, to prove the right hand side inequality we put as
above $D=d^2$, $d\geq 1$. Then  $D<N^2p^2$, whence $d<Np$. Taking
into account integrity of $d$ we obtain $d\leq Np-1$, or what is
the same
\begin{equation*}
\label{equ:N=q^2-2}
0\leq (Np-1)^2-D=2N(q^2-p)+1-q^2,
\end{equation*}
hence
\begin{equation}\label{equ:New}
q^2-p\geq (q^2-1)/2N\geq 0.
\end{equation}
Therefore, $p<q^2$ and by virtue of integrity of $q$ and $p$ we
arrive at a stronger inequality
$$
q+1\leq p\leq q^2-1
$$
and the theorem is proved.
  \qed \end{proof}

\begin{cor}\label{corol:int}
If \ $\g$, $|\g|>1$, is an integer number  then (\ref{main0}) has
no (non-trivial) algebraic $N$-solutions.
\end{cor}

\begin{prop}\label{corol:full}
The following representation holds
\begin{equation}\label{equ:gamma-param}
\mathcal{A}=\left\{\frac{2N-1+\t ^2}{2\t N} \;: \quad \t
\in\Q{}\cap(0,1),\; N\in\N{}\right\}.
\end{equation}
\end{prop}

\begin{proof}
We notice that the discriminant $D$ in (\ref{equ:discr}) will be
the a squared integer if and only if a positive integer valued
solution $(x,y)$ to the following Pell type equation
\begin{equation}\label{equ:x-y}
N^2x^2-(2N-1)y^2=1,
\end{equation}
does exist such that $x>y$. Here $x=p/d$, $y=q/d$ and $\g$ is
uniquely defined by $\g=x/y$.

Then (\ref{equ:x-y}) can be resolved by the standard
rationalization technique. Really, let $x=(1+\t y)/N$, whence $\t
$ should be a rational number. Substituting the last expression in
(\ref{equ:x-y}) gives
$$
y=\frac{2\t }{2N-1-\t ^2}, \qquad x=\frac{2N-1+\t ^2}{N(2N-1-\t
^2)}.
$$
Thus, we have
$$
\g=\frac{2N-1+\t ^2}{2\t N}.
$$

To define the admissible values of  $\t >0$ which provide  $\g>1$
we notice that
$$
\g-1=\frac{(\t -2N+1)(\t -1)}{2\t N},
$$
whence $\t \in(0;1)\cup(2N-1;+\infty)$. We can reduce our
representation to the interval $(0;1)$ only since  the invariance
property of (\ref{equ:gamma-param}) with respect to involution $\t
\to (2N-1)/\t $.
 \qed \end{proof}

\section{Maximal and minimal series}\label{sub:max} Now we study
the set $\mathcal{A}$ in dependence of the fractional
de\-com\-po\-si\-tion $\g=p/q$. We suppose as before that  $p$ and
$q$ have no common divisors.

\begin{prop}[Maximal series]\label{cor:p^2<q}
For $p/q\in\mathcal{A}$ the following sharp estimates hold
\begin{equation}\label{equ:max1}
p\leq q^2-2
\end{equation}
with equality only for odd denominator $q=2s+1$, $p=4(s^2+s)-1$,
and $N=s(s+1)$, where $s$ is a positive integer. Moreover, if the
denominator $q=2s$ is even we have a stronger inequality
\begin{equation}\label{equ:max2}
p\leq \frac{q^2-2}{2}=2s^2-1
\end{equation}
with equality iff
\begin{equation*}\label{equ:mmmm}
p=\frac{q^2-2}{2},\qquad N=\frac{q^2-4}{4}.
\end{equation*}
\end{prop}

\begin{proof}
First prove (\ref{equ:max1}). In view of (\ref{e-ii}), it suffices
to exclude the case $q=p^2-1$. Assuming the contrary, and let the
last equality hold. Then (\ref{equ:New}) and (\ref{equ:finite})
together  yield
$$
q^2\leq 2N<q^2\frac{(p^2+2-q^2)}{p^2},
$$
which easily implies $q^2\leq 2$ that contradicts to the lower
bound (\ref{e-i}).

In order to analyze the equality case, assume that $p=q^2-2$. It
follows then from (\ref{equ:New}) that $4N\geq q^2-1=p+1$ and
(\ref{equ:discr}) can be rewritten as follows
\begin{equation*}\label{equ:d1}
d^2=N^2p^2-(2N-1)(p+2)=(Np-1)^2-(4N-1-p)\leq (Np-1)^2
\end{equation*}
where $d>0$ is an integer. Hence, $d\leq (Np-1)$.

On the other hand,
$$
d^2=(Np-2)^2+(p-2)(2N-1)> (Np-2)^2
$$
which together with the preceding inequality implies $d=Np-1$, and
consequently
$$
p=q^2-2=4N-1.
$$
In particular, $q$ must be an odd number. One can readily find
that in this case  $N$ is an integer. Thus, the first case of the
corollary is proved.

To prove the second statement we suppose $q=2s$. We have from
(\ref{e-i}) $s\geq 2$ and due to irreducibility of $p/q$ we notice
that $p$ should be an odd number.

By Theorem~\ref{theo:Q-P} we have for the discriminant
$$
D\equiv d^2=N^2p^2-4s^2(2N-1).
$$
The last identity shows that $d$ has the same parity as $Np$ does.
Thus, $ d\leq pN-2.
$
On the other hand,
$$
(Np-2)^2\geq d^2=(Np-2)^2+4(Np-1-s^2(2N-1))
$$
and therefore,
\begin{equation*}\label{equ:y}
p-2s^2\leq -\frac{s^2-1}{N}.
\end{equation*}
The last inequality and $s\geq 2$ yield
$$
p<2s^2=\frac{q^2}{2}.
$$
Now, the inequality (\ref{equ:max2}) follows from the evenness of
$q$.

To analyze the equality case in (\ref{equ:max2}) we observe that
in our notation $d>Np-3$, otherwise we would have
$$
(Np-3)^2\leq d^2=N^2p^2-2(2N-1)(p+1),
$$
which implies $ 2N(p-2)\leq 7-2p$, and contradiction follows.
Thus, $d=Np-2$ and a straightforward  computation shows
$$
N=\frac{p-1}{2}=s^2-1
$$
which completes the proof.   \qed \end{proof}

\begin{prop}[Minimal series]\label{corol:p=q+2}
For $p/q\in\mathcal{A}$  the following lower bound holds
\begin{equation}\label{equ:min}
p\geq q+2
\end{equation}
with equality iff  $q=2s+1$ is an odd number, $p=q+2$ and
$N=(q-1)/2$.
\end{prop}

\begin{proof}
Form (\ref{e-ii}) we have $p\geq q+1$. Again, we argue by
contradictory and assume that $p=q+1$. Then we have the following
relation for the discriminant
\begin{equation}\label{equ:d^2}
D=N^2(q+1)^2-q^2(2N-1)=d^2,
\end{equation}
and $d$ is a positive integer. Hence
$$
d^2=q^2(N-1)^2+2qN^2+N^2=\left(q(N-1)+\frac{N^2}{N-1}\right)^2-\frac{(2N-1)N^2}{(N-1)^2}.
$$
In particular,
$$
d<q(N-1)+\frac{N^2}{N-1}=q(N-1)+N+1+\frac{1}{N-1}
$$
which implies
$$
d\leq q(N-1)+N+1\equiv d_1+1.
$$

On the other hand
$$
d_1^2= (q(N-1)+N)^2=q^2(N-1)^2+2qN(N-1)+N^2<d^2,
$$
and it follows that $ d=d_1+1. $ Then by virtue of (\ref{equ:d^2})
we arrive at $ 2q=2N-1. $ Thus, the contradictions shows that
$p\geq q+2$ and (\ref{equ:min}) is proved.

Now assume that the equality $p=q+2$ holds. Arguing as above we
obtain the following inequality
\begin{equation}\label{strong}
d<(N-1)q+2N+2+\frac{2}{N-1}\leq (N-1)q+2N+3.
\end{equation}
On the other hand,
\begin{equation}\label{equ:even}
d^2\equiv q^2(N-1)^2+4qN^2+4N^2> ((N-1)q+2N)^2,
\end{equation}
which implies by strong inequality in (\ref{strong})
$$
(N-1)q+2N+1\leq d\leq(N-1)q+2N+2.
$$
But it follows from (\ref{equ:even}) that $d$ and $(N-1)q$ have
the same parity. This gives exactly one choice in the last
inequality
$$
d=(N-1)q+2N+2
$$
which after comparison with the  definition of $d$ in
(\ref{equ:d^2}) implies $ q=2N+1$, and the required relation
follows. \qed \end{proof}

The `algebraic' constants $\gamma\in\mathcal{A}$ for small
denominators $q\leq 30$ are displayed in Figure~\ref{fig:p-q}.

\medskip

\begin{figure}[h]
  \begin{center}
        \includegraphics[width=0.75\textwidth]{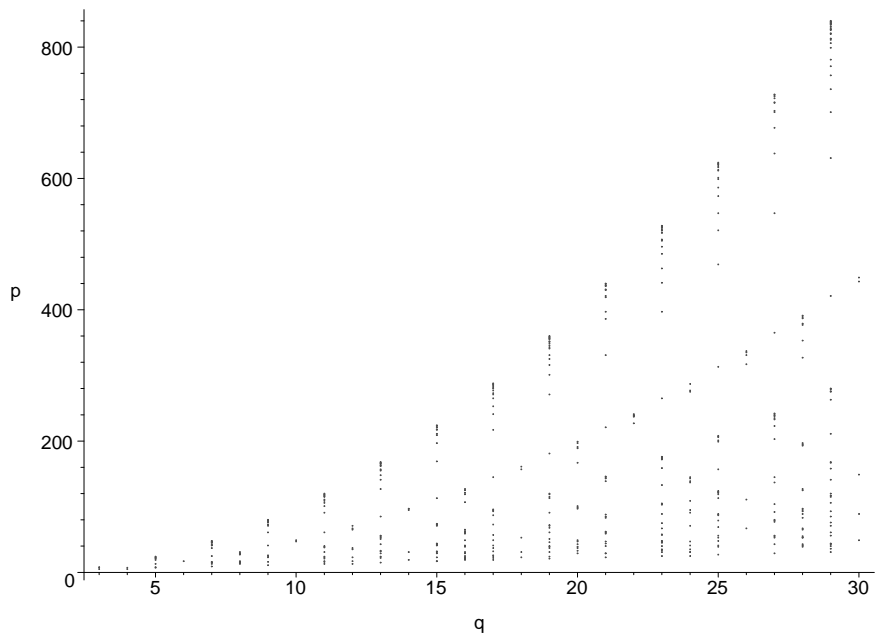}
        \caption{$q$-$p$ Diagram}
    \label{fig:p-q}
  \end{center}
\end{figure}

\section{Concluding remarks}

We wish point out that our formula (\ref{equ:poly}) can also be
deduced by using the general representation of $p$-harmonic
functions in the plane near their singular points. The mentioned
representation is established in \cite{IM} and \cite{M} by using the
hodograph method. Our approach, nevertheless, is more direct and use
no the quasiregularity of the complex gradient of the $N$-solutions.

The author wish to thank  the anonymous reviewer for many
constructive advices and bringing our attention to the papers
\cite{IM} and \cite{M}.

%\begin{acknowledgement}
%The author is grateful to G\"unnar Aronsson, Bj\"orn Gustafsson and
%Henrik Shahgholian  for many fruitful discussions.
%\end{acknowledgement}

\end{document}